\documentclass[sigconf]{acmart}

\AtBeginDocument{%
  \providecommand\BibTeX{{%
    \normalfont B\kern-0.5em{\scshape i\kern-0.25em b}\kern-0.8em\TeX}}}
    
\setcopyright{acmcopyright}
\copyrightyear{2022}
\acmYear{2022}
\setcopyright{acmcopyright}
\acmConference[MobiHoc '22]{The Twenty-third International Symposium on Theory, Algorithmic Foundations, and Protocol Design for Mobile Networks and Mobile Computing}{October 17--20, 2022}{Seoul, Republic of Korea}
\acmBooktitle{The Twenty-third International Symposium on Theory, Algorithmic Foundations, and Protocol Design for Mobile Networks and Mobile Computing (MobiHoc '22), October 17--20, 2022, Seoul, Republic of Korea}
\acmPrice{15.00}
\acmDOI{10.1145/3492866.3549716}
\acmISBN{978-1-4503-9165-8/22/10}

\usepackage{microtype}
\usepackage{graphicx}
\usepackage{booktabs} 
\usepackage{mathtools}
\usepackage{dsfont}
\usepackage[english]{babel}
\usepackage{amsthm}
\usepackage[ruled,vlined,linesnumbered,noend]{algorithm2e}
\usepackage{algpseudocode}
\usepackage{subcaption}
\usepackage{soul}

\newtheorem{remark}{Remark}

\newcommand{\g}{\mathbf{g}}

\newcommand{\w}{\mathbf{w}}

\newcommand{\y}{\mathbf{y}}

\newtheorem{thm}{Theorem}

\newtheorem{lem}[thm]{Lemma}

\allowdisplaybreaks[1]
\begin{document}

\title{On Scheduling Ring-All-Reduce Learning Jobs in Multi-Tenant GPU Clusters with Communication Contention}

\author{Menglu Yu,$^{1}$ Bo Ji,$^{2}$ Hridesh Rajan,$^{1}$ and Jia Liu$^{3,1}$
}
\affiliation{
	\institution{$^{1}$Department of Computer Science, Iowa State University}
	\institution{$^{2}$Department of Computer Science, Virginia Tech}
	\institution{$^{3}$Department of Electrical and Computer Engineering, The Ohio State University}
	\country{}
}

\renewcommand{\shortauthors}{Yu and Liu, et al.}


\begin{abstract}
Powered by advances in deep learning (DL) techniques, machine learning and artificial intelligence have achieved astonishing successes.
However, the rapidly growing needs for DL also led to communication- and resource-intensive distributed training jobs for large-scale DL training, which are typically deployed over GPU clusters.
To sustain the ever-increasing demand for DL training, the so-called ``ring-all-reduce'' (RAR) technologies have recently emerged as a favorable computing architecture to efficiently process network communication and computation load in GPU clusters.
The most salient feature of RAR is that it removes the need for dedicated parameter servers, thus alleviating the potential communication bottleneck.
However, when multiple RAR-based DL training jobs are deployed over GPU clusters, communication bottlenecks could still occur due to contentions between DL training jobs. 
So far, there remains a lack of theoretical understanding on how to design contention-aware resource scheduling algorithms for RAR-based DL training jobs, which motivates us to fill this gap in this work.
Our main contributions are three-fold: i) We develop a new analytical model that characterizes both communication overhead related to the worker distribution of the job and communication contention related to the co-location of different jobs; 
ii) Based on the proposed analytical model, we formulate the problem as a non-convex integer program to minimize the makespan of all RAR-based DL training jobs.
To address the unique structure in this problem that is not amenable for optimization algorithm design, we reformulate the problem into an integer linear program that enables provable approximation algorithm design called SJF-BCO (\ul{S}mallest \ul{J}ob \ul{F}irst with \ul{B}alanced \ul{C}ontention and \ul{O}verhead); 
and iii) We conduct extensive experiments to show the superiority of SJF-BCO over existing schedulers.
Collectively, our results contribute to the state-of-the-art of distributed GPU system optimization and algorithm design.
\end{abstract}



\maketitle


\section{Introduction} \label{sec:intro}

{\bf Background and Motivation:} 
In recent years, the rise of deep learning has driven an ever-increasing need for large-scale distributed training in GPU clusters, which leverages massive parallelism to speed up the training processes.
This has been evidenced by the popularity of several prevailing distributed deep learning (DDL) frameworks (e.g., TensorFlow~\cite{TensorFlow} and PyTorch~\cite{Paszke19:PyTorch}).
In these DDL frameworks, the traditional and most widely adopted computing-networking structure is based on the sever-worker (SW) architecture, where DDL training jobs are decomposed into and executed in parallel by a set of workers under the coordination of a parameter server.
However, 
as the number of workers increases, the SW architecture suffers from serious scalability limitations since the server acts as a communication bottleneck and a single-point-of-failure.
To address the scalability limitations of the SW architecture, the ring-all-reduce (RAR)~\cite{Patarasuk09:Ring-AllReduce} architecture has attracted increasing attention in recent years.
The key idea of RAR is that, by forming a ring and working collaboratively, the workers can update the learning model parameters {\em without} needing any parameter server, thus removing the communication bottleneck and alleviating the single point of failure.
Moreover, it can be shown that the RAR architecture enjoys the highly desirable ``bandwidth optimality'' in the sense that, as the number of workers increases, the amount of information exchanged in the network is asymptotically {\em independent} of the number of workers (see Section~\ref{sec:Preliminaries} for details). 

However, despite all these salient features, the performance of deploying RAR-based training jobs in {\em multi-tenant} GPU clusters remains far from being satisfactory in practice~\cite{Wang20:contention}.
The fundamental reason is that the bandwidth optimality of RAR architecture only happens when there is only a single training job in the system (i.e.,  a contention-free environment).
In a multi-tenant GPU cluster, however, such an ideal contention-free condition is rarely satisfied.
As a result, significant communication bottleneck links could occur when deploying RAR-based training jobs in the system.
For example, researchers in~\cite{Wang20:contention} have found that on a cluster of four-GPU servers connected by 10 Gbps Ethernet, when only one RAR training job is executed with four GPUs in the cluster, the job completion time is 295 seconds.
In comparison, when four jobs of the same type are executed simultaneously with each job still using four GPUs but scheduled across GPU servers, each job's completion time dramatically increases to 675 seconds due to the extensive communication contention.
These empirical performance results of RAR indicate that developing efficient and effective scheduling for RAR-based DDL training jobs is well warranted to mitigate contention-induced communication bottlenecks.
However, in the literature so far, there remains a lack of theoretical understanding on how to design contention-aware resource scheduling algorithms for RAR-based DDL training jobs.
In light of the rapidly growing importance of RAR-based DDL deployment, our goal in this paper is to fill this gap and develop contention-aware scheduling algorithms for RAR-based training jobs in multi-tenant GPU clusters.

%

\smallskip
{\bf Technical Challenges:}
We note, however, that due to a number of technical difficulties, developing contention-aware scheduling algorithms for RAR-based DDL jobs in multi-tenant GPU clusters is highly challenging.
\emph{First and foremost}, just as any network optimization problems that deal with contentions and interferences, the completion time of an RAR-based training job depends not only on its resource allocation decisions (i.e., the number of ring-forming workers and their locality), but also on the number of {\em concurrent} RAR-based DDL jobs that (partially or completely) share the communication links of this job.
The complex communication coupling between concurrent RAR-based training jobs renders it {\em intractable} to compute the per-iteration execution time of an RAR-based DDL job in closed-form.
\emph{Second}, there exists a fundamental trade-off in terms of job locality.
On one hand, co-locating all workers of an RAR-based DDL job on the same server enjoys a faster intra-server communication speed, but could lead to resource fragmentation. 
On the other hand, spreading the ring of an RAR job over multiple servers could also result in more contentions of communication links and overhead in establishing connections between servers.
\emph{Last but not least}, due to the resource constraints of each server and the iterative nature of DDL training workload, the resource allocation decision for each RAR-based training job is subject to a mix of packing and covering types of constraints, both of which are known to be NP-hard.

\smallskip
{\bf Our Contributions:}
In this paper, we overcome the above challenges and design a suite of scheduling algorithmic techniques for efficient RAR-based DDL training in multi-tenant GPU clusters with theoretical makespan performance guarantees.
The key idea of our algorithmic design is to transfer the structural complexity of the intractable  per-iteration running evaluation in the original scheduling problem to the dimensional complexity of an equivalent reformulated problem, which has a much cleaner integer linear program structure to work with.
Our main results and technical contributions are summarized as follows:

\begin{list}{\labelitemi}{\leftmargin=1.5em \itemindent=-0.0em \itemsep=.2em}
\item We first propose a new analytical framework for RAR-based DDL training resource allocation and scheduling that characterizes both communication contention and overhead under the RAR architecture in a multi-tenant GPU cluster.
This analytical modeling serves as the foundation to enable us to formulate the scheduling optimization framework to minimize the makespan of all RAR-based training jobs.

\item As mentioned earlier, due to the complex resource contentions and couplings between RAR-based DDL jobs, it is intractable to determine the closed-form expression for the per-iteration execution time for each DDL job.
To address this challenge, we further reformulate the original problem into an equivalent integer problem, which has a cleaner problem structure.
Doing so allows us to convert the structural complexity of the original problem to the exponential dimensionality complexity in the reformulated problem, which is more amenable for low-complexity search-based optimization algorithm design.

\item Based on the above problem reformulation, we propose an efficient scheduling algorithm called SJF-BCO (\ul{s}mallest \ul{j}ob \ul{f}irst with \ul{b}alanced \ul{c}ontention and \ul{o}verhead) with theoretical approximation ratio guarantee.
We conduct extensive experiments to verify the performance of our proposed SJF-BCO algorithm and compare with existing scheduling policies to show the superiority of SJF-BCO over these baselines.
\end{list}
Collectively, our results contribute to a comprehensive and fundamental understanding of RAR-based DDL resource scheduling optimization.
The roadmap of the rest of the paper is as follows.
In Section~\ref{sec:Related}, we review the related literature.
Section~\ref{sec:Preliminaries} present preliminaries to familiar readers with the necessary background.
Section~\ref{sec:model_formulation} introduces the system model and problem formulation.
Section~\ref{sec:alg} demonstrates our algorithms and Section~\ref{sec:Performance} provides their performance analysis.
Section~\ref{sec:numerical} presents numerical results and Section~\ref{sec:conclusion} concludes this paper.


\section{Related Work} \label{sec:Related}

As mentioned in Section~\ref{sec:intro}, DDL training job scheduling algorithms have received growing interest recently.
Research in this area aims to schedule DDL jobs and manage computing resources efficiently in multi-tenant GPU computing clusters.
Early attempts in this field were mostly heuristic approaches based on empirical observations and models to conduct the resource scheduling (e.g.,~\cite{Gu19:Tiresias,Mahajan20:Themis,Mei17:energy,Chau17:energy}). 
For example, Gandiva~\cite{Xiao18:Gandiva} considered GPU time-slicing and job scheduling by predicting DDL training jobs characteristics.
Optimus~\cite{Peng18:Optimus} leveraged performance models through online-fitting to guide the job scheduling aiming to minimize training completion time.
Rather than using prediction models, another line of research is to take advantage of the model-less data-riven learning methods for DDL job scheduling (e.g.,~\cite{Bao19:Harmony,Hu19:Spear,Wang20:Metis}).
For instance, Harmony~\cite{Bao19:Harmony}, a deep-reinforcement-learning-based  scheduler considered minimizing the job completion time.
Hu.~{\it et al.}~\cite{Hu19:Spear} designed a new scheduling framework called Spear to minimize the makespan of jobs by leveraging the deep reinforcement learning techniques.
However, these works do not provide theoretical performance guarantee.
Also, none of these works considered RAR-based DDL job scheduling.

The most related work to this paper is GADGET~\cite{Yu22:GAGET}, which characterized RAR-based DDL job scheduling based on the assumption that the bandwidth of each job is reserved.
As a result, there is no need to consider communication contention in~\cite{Yu22:GAGET}. We note that a limitation of the reserved bandwidth assumption is that it could lead to resource under-utilization.
In contrast, this paper considers communication contention to avoid this limitation.
This, however, renders the scheduling problem far more challenging.
Lastly, Wang et al.~\cite{Wang20:contention} also considered contention under various all-reduce architectures, including RAR.
However, they also relied on a system-dependent online-fitting model to predict the execution time and did not explicitly formulate any scheduling optimization problem.
Their solution was based on heuristics without theoretical performance guarantee.
In contrast, we develop an analytical model to facilitate the job scheduling as a rigorous optimization problem, which in turn entails approximation algorithm design with theoretical performance guarantee.
\section{Ring-All-Reduce (RAR)-Based Distributed Learning: A Primer}\label{sec:Preliminaries}

In this section, we  provide a quick overview on the RAR-based distributed learning 
to familiarize readers with necessary background and fix the terminologies that are useful in the rest of the paper.

\begin{figure}[t!]
\includegraphics[width=.48\textwidth]{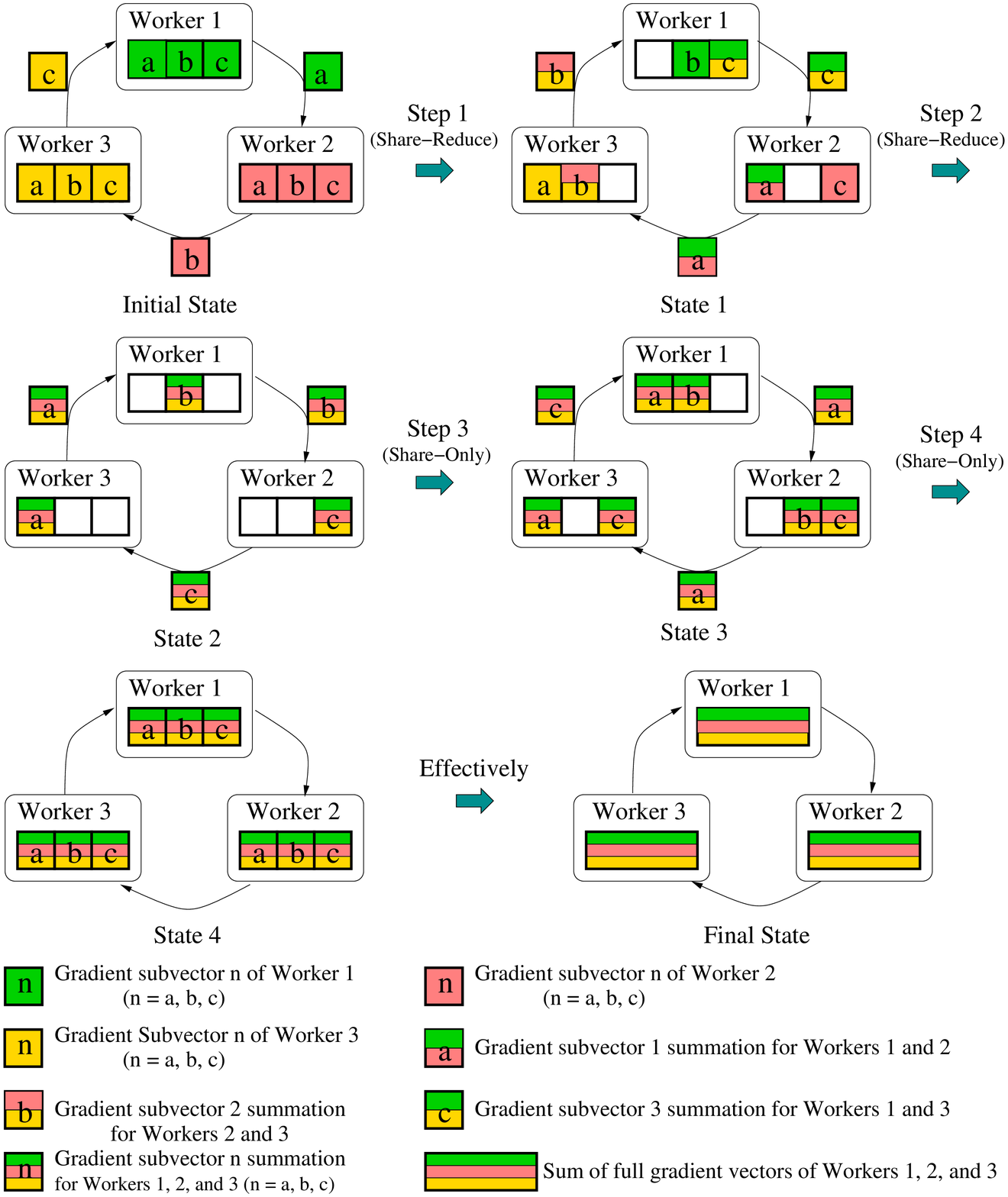}
\caption{A three-worker illustrative example of the ring-all-reduce (RAR) process.}
\label{fig:ringallreduce}
\vspace{-.1in}
\end{figure}


\smallskip
{\bf 1) SGD-Based Distributed Learning:}
The training of many ML problems is typically in the form of an empirical risk minimization (ERM) problem: $\min_{\w \in\mathbb{R}^d} \bar{L}(\w)\triangleq\frac{1}{P} \sum_{i=1}^{P} L(\w,\delta_i)$, where $\w$ contains the model parameters to be learned, $L(\w,\delta_i)$ is a loss function, 
and $P$ is the total number of samples.
Due to the high-dimensionality and the large dataset size of many ERM problems (e.g., in deep learning), the stochastic gradient descent (SGD) method has become the most widely adopted method.
The SGD method can be written as the following iterative process: $\w_{k+1} = \w_k  -  (\eta_k/Q)\sum_{i\in\mathcal{Q}_k} \g_k^i$, where $\eta_k$ denotes the learning rate in the $k$-th iteration, $\mathcal{Q}_k$ represents the mini-batch in the $k$-th iteration with $|\mathcal{Q}_k| =Q$, and $\g_k^i$ is a stochastic gradient based on a random sample $\delta_i \in \mathcal{Q}_k$.
The finite-sum and mini-batch structure of SGD naturally lends itself to a {\em distributed} implementation in a $Q$-worker DDL system coordinated by a parameter server as follows:
First, the dataset is partitioned by $Q$ workers.
In each iteration $k$, each worker retrieves the current model parameters from the server and randomly draws a sample from its local dataset, and then computes a stochastic gradient (e.g., using the backpropagation method).
Then, all workers send their gradients to the server to be aggregated.

\smallskip
{\bf 2) The Ring-All-Reduce (RAR) Architecture:}
It can be seen from the above discussions that SGD-based distributed learning naturally implies a server-worker (SW) architecture.
However, as mentioned in Section~\ref{sec:intro}, the SW architecture suffers from scalability limitations as the number of workers increases.
This is because all workers need to communicate with the server, which creates a bottleneck.
Specifically, a $w$-worker SW system that solves a $d$-dimensional ERM problem requires $2wd$ amount of data exchange per iteration (each worker sends and receives a $d$-dimensional vectors per iteration), which scales {\em linearly} with respect to $w$.

To address this scalability limitation, the RAR~\cite{Patarasuk09:Ring-AllReduce} has been proposed to {\em remove the server}.
Under RAR, the workers form a ring to exchange and aggregate data collaboratively. 
For a $w$-worker RAR system, each worker splits its stochastic gradient into $w$ sub-vectors (see Fig.~\ref{fig:ringallreduce} for an example with $w=3$).
Each iteration of RAR has $2(w-1)$ steps that can be divided into two phases.
In the first phase (steps $1,\ldots,w-1$), workers perform gradients reduction (i.e., summation), where each worker receives a gradient subvector from its upstream worker and sends its local reduction result to its downstream worker (Share-Reduce phase).
In the second phase (steps $w,\ldots,2w-2$), each worker circulates its resultant sub-vectors following the same token to obtain its final resultant gradients vector (Share-Only phase).
Since each worker sends $\frac{d}{w}$ amount of data for $2(w-1)$ times, the total amount of data any worker receives is $\frac{2d(w-1)}{w}$, which is asymptotically {\em independent} of $w$ as $w$ increases (also referred to as being {\em bandwidth-optimal} in the literature).


\section{System Model and Problem Formulation} \label{sec:model_formulation}

In this section, we first introduce our system model in Section~\ref{subsec:sys_model} and then present the problem formulation for RAR-based DDL scheduling optimization in multi-tenant GPU clusters in Section~\ref{subsec:formulation}.

\subsection{System Model} \label{subsec:sys_model}

{\bf 1) Scheduling Model:}
Consider a multi-tenant GPU cluster that contains a set of servers $\mathcal{S}$.
Each server is equipped with a set of homogeneous (i.e., of equal computation speed) and synchronized GPUs.
The servers in $\mathcal{S}$ are connected by a network and the network topology can be modeled as a connected graph.
In the beginning of a scheduling horizon $\mathcal{T}$ of length $|\mathcal{T}| = T$ time-slots, there is a set of RAR-based DDL jobs $\mathcal{J}$ waiting to be scheduled for training over $\mathcal{T}$. 
Each job $j \in \mathcal{J}$ is associated with a number of GPUs $G_j$ and a total number of training iterations $F_j$ from its users, both of which are requested by its users.\footnote{In practice, to prevent spending excessively long time waiting for the training process of a DDL job to converge, a maximum number of training iterations is usually given.}

In this paper, we consider the ``gang-scheduling'' discipline that is widely adopted in practical large-scale GPU clusters~\cite{Mahajan20:Themis,Gu19:Tiresias,Wang20:contention}.
Under gang scheduling, all workers (i.e., GPUs) of an RAR-based DDL job should be allocated {\em simultaneously.}
Moreover, once a job is scheduled to start, all GPUs allocated for this job will run to the job's completion and no preemption/migration is allowed.\footnote{Besides the overhead and complication added for both software and hardware, it has been shown that  frequent job preemption and migration may lead to significant performance degradation~\cite{Gu19:Tiresias}.}
Upon the job's completion, the occupied resource will also be released simultaneously.
Each GPU can only be occupied by one worker of some job at any given time.
As shown in Fig.~\ref{fig:example}, the workers of a job can be allocated within a single server or across multiple servers, as long as there exists a path in the underlying network that connects these workers and forms a ring topology to perform the RAR process. 
Note that Fig.~\ref{fig:example}(a) allocates the workers in the same server for each job, thus having no communication overhead.
On the contrary, Fig.~\ref{fig:example}(b) allocates workers across different servers for each job, which introduces communication contention when the two jobs happen to perform RAR communication concurrently.

\begin{figure}[t!]
    \centering
    \includegraphics[width=0.45\textwidth]{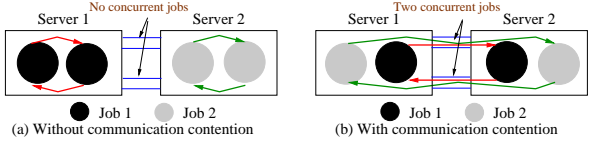}
    \caption{An example of worker placement.}
    \label{fig:example}
    \vspace{-.1in}
\end{figure}

In this system, the control decisions of the scheduler are: i) determine a feasible scheduling for all jobs in $\mathcal{J}$ subject to network resource capacity; and ii) determine each job's starting time.
Specifically, consider an RAR-based DDL job $j$ scheduled with $w_j$ workers and its gradient size is $m_j$.
Let $y_{js}[t] \in \mathbb{Z}^{+}$ denote the number of GPUs scheduled for job $j$ on server $s$ in time-slot $t \in \mathcal{T}$.
Then, a scheduling decision in time-slot $t$ can be fully defined by the vector $\y[t] \triangleq [y_{js}[t], \forall j, s]$.
Let $a_j=\arg\min_t\{y_{js}[t]>0, \forall s\}$ be the starting time of job $j$ (to be determined) by the scheduling and let $T_j$ be the resultant completion time of job $j$.
Let $\mathcal{J}[t]\triangleq\{j | t\in[a_j,T_j]\}$ represent the set of {\it active jobs} (jobs being executed) in time slot $t$.
Clearly, to satisfy the $G_j$ number of GPUs requested for job $j$ during its training time,
we have:
\begin{align}
    \label{eqn:gpu}&\sum_{s\in\mathcal{S}}y_{js}[t] = G_j, \quad \forall j\in\mathcal{J}[t], t\in\mathcal{T}.
\end{align}
Also, scheduling decisions $\y[t]$, $\forall t$ are subject to GPU resource constraints.
Let $O_s$ represent the GPU capacity of server $s$.
To ensure that the allocated GPUs do not exceed each server's limit, we have:
\begin{align}
    \label{eqn:cap}&\sum_{j\in\mathcal{J}[t]}y_{js}[t]\leq O_s, \quad\forall s\in\mathcal{S}, t\in\mathcal{T}.
\end{align}
Also, under the non-preemptive gang scheduling, we have:
\begin{align}
y_{js}[t] = y_{js}[t-1], \quad\forall s\in\mathcal{S}, j\in\mathcal{J}[t], a_j < t\leq T_j.
\end{align}
Finally, to ensure that no workers should be allocated for non-active jobs and positive integer number of workers should be assigned to active jobs, we have:
\begin{align}
&y_{js}[t] = 0, \quad \forall s\in\mathcal{S}, j\not\in\mathcal{J}[t], t\in\mathcal{T}, \\    
\label{eqn:int}&y_{js}[t]\in\mathbb{Z}^{++},  \quad \forall s\in\mathcal{S}, j\in\mathcal{J}[t], t\in\mathcal{T}.
\end{align}

\smallskip
{\bf 2) Communication Contention Modeling:} 
With the above scheduling model, we are now in a position to present our communication contention model.
We assume that no communication contention will be introduced if at most one server is used for the job.
For example, in Fig.~\ref{fig:example}(a), jobs 1 and 2 both use intra-server communication and does not incur any communication contention.
By contrast, in Fig.~\ref{fig:example}(b), jobs 1 and 2 induce communication contention since they both compete for inter-server link bandwidth between servers 1 and 2.
We let $p_j[t]$ denote the largest number of concurrently running jobs that share an inter-server communication link with job $j$ in time slot $t$, which can be computed as:
\begin{align}
\label{eqn:p_j}p_j[t] &=\max_{s\in\mathcal{S}}\bigg\{\mathds{1}\{0<y_{js}[t]<G_j\}\sum_{j'\in\mathcal{J}[t]}\mathds{1}\{0\!<\!y_{j's}[t]<G_{j'}\}\bigg\}, \nonumber\\
&\hspace{1.7in}\forall j\in\mathcal{J}[t], t\in\mathcal{T}.
\end{align}
In \eqref{eqn:p_j}, the first term $\mathds{1}\{0<y_{js}[t]<G_j\}$ indicates that only active jobs using inter-server communication on server $s$ will be considered.
The second term $\sum_{j'\in\mathcal{J}[t]}\mathds{1}\{ 0<y_{j's}[t]<G_{j'}\}$ represents the number of different jobs that compete for  inter-server communication on server $s$.
Since job $j$ may not be transmitting at all times (due to switching between communication and computation modes), we let $k_j[t]$ be the actual largest number of {\em contending} jobs on average with job $j$ in time-slot $t$, which can be assumed to be statistically linearly proportional to $p_j[t]$, i.e.,
\begin{align}
k_j[t] = \xi_1 p_j[t], \quad\forall j\in\mathcal{J}[t], t\in\mathcal{T},
\end{align}
where $\xi_1\in(0,1]$ is a positive constant.

\smallskip
{\bf 3) RAR-Based DDL Training Completion Time Modeling:}
To evaluate the job completion time $T_j$ of job $j$, we need to first characterize the RAR training speed.
Note that the per-iteration RAR operation time of each DDL job can be decomposed into three parts: i) information exchange time, ii) computation time, and iii) communication overhead.
Next, we will model the operation time of each part individually.

\smallskip
{\em 2-1) Information Exchange Time:}
We use $B^{\{\omega_{j,1},\omega_{j,2}\}}(\mathbf{y}[t])$ to denote the bandwidth between two successive workers $\omega_{j,1}$ and $\omega_{j,2}$ in job $j$'s ring in time-slot $t$ under a scheduling decision $\y[t]$, where $\omega_{j,2}$ is the downstream worker of $\omega_{j,1}$.
Note that, unlike~\cite{Yu22:GAGET}, we do not reserve bandwidth for each job in this paper, and this bandwidth is determined by communication contention with other jobs under the scheduling decisions $\y[t]$ (see Fig.~\ref{fig:example}(b)).
We let $B_j(\y[t])\triangleq\min_{(\omega_{j,1},\omega_{j,2})\in\mathcal{L}_j} B^{\{\omega_{j,1},\omega_{j,2} \}}(\mathbf{y}[t])$ represent the bandwidth of the {\em bottleneck} link of job $j$ under scheduling decision $\y[t]$, where $\mathcal{L}_j$ is the set of all links  of job $j$.
Recall from Section~\ref{sec:Preliminaries} that the amount of information exchanged in each time-slot can be computed as $\frac{2m_j}{w_j}(w_j-1)$.
Thus, the number of time-slots for information exchange can be computed as $\frac{2m_j}{w_j}(w_j-1)/B_j(\y[t])$.

Clearly, the bottleneck link of job $j$ occurs in those links that are shared by the largest number of other concurrently running jobs.
We let $b^e$ and $b^i$ be the link bandwidth between and within servers, respectively, where $b^i\gg b^e$ in practice~\cite{Zhang17:Poseidon,Shi18:DDL}.
Recall that $k_j[t]$ denotes the actual {\it largest} number of contending jobs on average with job $j$ in time-slot $t$.
Ideally, each job on this bottleneck link has an equal share of bandwidth $b^e/k_j[t]$ under communication contention.
In practice, however, the bandwidth performance often degrades when multiple jobs compete for a link, which results in each job having less than $b^e/k_j[t]$ share of bandwidth if $k_j[t]\geq 2$~\cite{Wang20:contention}.
To model this effect, we use a function $f(\boldsymbol{\alpha}, k_j[t])$ to represent the ``bandwidth sharing degradation factor'' under communication contention, where $\boldsymbol{\alpha}\in\mathbb{R}^d$ captures all parameters that could lead to degradation.
We assume that $f(\boldsymbol{\alpha}, k_j[t])$ satisfies the following properties: i) $f(\boldsymbol{\alpha}, 1) = 1$ and ii) $f(\boldsymbol{\alpha}, k_j[t])$ is an increasing function of $k_j[t]$.
For example, if $f(\boldsymbol{\alpha}, k_j[t])$ is a linear function $k_j[t] + \alpha(k_j[t]-1)$, then $B_j(\y[t])= b^e/f(\boldsymbol{\alpha}, k_j[t])  = b^e/( k_j[t] + \alpha(k_j[t] -1))$.

Recall that in the special case where all workers of a job $j$ are co-located within a single server, there is no contention.
Further, intra-server communication is typically enabled by fast interconnect techniques (e.g., NVLink~\cite{Foley17:NVLink}). 
Hence, we have $B_j(\y[t])=b^i$. 

\smallskip
{\em 2-2) Computation Time:}
To characterize the computation time in the RAR operation, we use $C$ to denote the computational speed of a GPU unit (defined as the amount of data processed in each time-slot).
Since there are $\frac{m_j}{w_j}(w_j-1)$ amount of data for reduction in each RAR operation, the number of time-slots to complete all reductions can be computed as $\frac{m_j}{w_j}(w_j-1)/C$.
In addition to the all-reduce operation time, the computation time also includes the forward pass (FP) time and the backward pass (BP) time to compute a stochastic gradient.
We let $\Delta^f_j$ ($\Delta^b_j$) denote the duration of one FP (BP) of job $j$.
Note that the FP time is proportional to the mini-batch size $M_j$, which can be calculated as $\Delta^f_jM_j$ (the size of a mini-batch multiplied by the FP processing time of one sample).
Meanwhile, the BP time $\Delta^b_j$ is usually not relevant to the mini-batch size $M_j$ and is typically fixed.
Thus, the total number of time-slots for per-iteration computation can be computed as $\frac{m_j}{w_j}(w_j-1)/C + \Delta^f_jM_j + \Delta^b_j$.

\smallskip
{\em 2-3) Communication Overhead:}
In practice, it has been observed that typically, the more servers an RAR-based DDL job uses to perform the training, the larger the latency due to communication overhead (e.g., ACK time for message transmission, negotiation time among all workers before conducting all-reduce~\cite{Sergeev18:Horovod}) can be~\cite{Wang20:contention}.
In this paper, we use  $\gamma_j(\mathbf{y}_j[t])$ to denote the latency of job $j$ caused by communication overhead in time-slot $t$.
We assume that the latency is linear proportional to the number of servers in use, i.e., $\gamma_j(\y_j[t])=\xi_2\sum_s\mathds{1}\{y_{js}[t] > 0\}$, where $\y_j[t]= [y_{js}[t]>0, \forall s]$ and $\xi_2\in(0, 1]$ is a positive constant.

Lastly, putting 2-1) -- 2-3) together, we can compute the RAR operation time of job $j$ under  scheduling decision $\y[t]$ as follows:
\begin{align} \label{contention_model}
\!\!\! \tau_j[t] \!=\!\frac{\frac{m_j}{w_j}\cdot2(w_j\!-\!1)}{B_j(\y[t])}\!+\!\frac{\frac{m_j}{w_j}\cdot(w_j\!-\!1)}{C}\!+\!\gamma_j(\y_j[t])\!+\!\Delta^f_jM_j + \Delta^b_j.
\end{align}
Hence, the RAR training speed $\phi_j[t]$ (i.e., the number of mini-batch iterations completed by job $j$) in time-slot $t$ can be computed as $\phi_j[t]\triangleq\lfloor (\tau_j[t])^{-1} \rfloor$.
Recall that $F_j$ is the requested number of iterations for training job $j$.
Thus, job $j$'s completion time can be calculated as:
\begin{align} \label{eqn:completion}
T_j=a_j+\arg\min_t \big\{\sum\nolimits_{t\in\mathcal{T}}\phi_j[t]\geq F_j \big\}, \quad\forall j\in\mathcal{J}[t].
\end{align}

\subsection{Problem Statement} \label{subsec:formulation}
In this paper, our goal is to determine the scheduling decisions $\y[t]$ to minimize the makespan (i.e., $\max_j T_j$), which is one of the most useful metrics to measure the efficiency of multi-tenant GPU clusters~\cite{Grandl16:OSDI,Grandl16-2:OSDI}.
Putting all modeling constraints and the objective together, the RAR-based DDL job scheduling problem (RAR-DDLS) can be formulated as the following optimization problem:
\begin{align*}
\text{\bf RAR-DDLS:} \underset{}{\ \min_{y_{js}[t], \forall j, s, t} } \max_{j\in\mathcal{J}} &\ T_j\\
    \text{subject to } 
    &\text{Constraints}~(\ref{eqn:gpu})-(\ref{eqn:completion}).  
\end{align*}

We note that Problem RAR-DDLS is an integer non-convex program with packing and covering constraints, which is NP-Hard.
In addition, the non-convex constraint in~(\ref{eqn:p_j}) involves indicator functions and the $\max$ operator, which cannot be written in a closed-from expression and hence is not amenable to conventional optimization techniques.
Due to these challenges, we will pursue an approximation algorithmic approach in Section~\ref{sec:alg} that offers provable approximation ratio guarantee.
To conclude this section, we summarize the key notations in this paper in Table~\ref{table:notation}.

\begin{table}[t!]
\centering
\caption{Notation.}
\label{table:notation}
\vspace{-.1in}
{\small
\begin{tabular}{| c | l|}
\hline
$\mathcal{T}/G_j$ &Scheduling time horizon/\# of GPUs requested by job j\\ \hline
$\mathcal{S}$/$\mathcal{N}$ &Set of servers/GPUs in the cluster\\ \hline
$\mathcal{J}[t]$ & The set of active jobs in time-slot t \\ \hline
$k_j[t]$ &  \begin{tabular}[c]{@{}l@{}l@{}}Actual largest number of contending jobs on average with\\ job j in time-slot $t$\end{tabular}\\ \hline
$\tau_j[t]$ &Per-iteration training time of job j in time-slot t\\ \hline
$y_{js}[t]$ &\# of GPUs scheduled on server $s$ for job j in time-slot $t$\\ \hline
$O_s$ &GPU capacity of server $s$\\ \hline
$a_j/T_j$ & Starting/completion time slot of job $j$ \\ \hline
$\mathcal{Y}$ &  \begin{tabular}[c]{@{}l@{}l@{}}The set of feasible scheduling schemes over $\mathcal{T}$\end{tabular}\\ \hline
$\y_j^k$ & \begin{tabular}[c]{@{}l@{}l@{}} A schedule of job $j$ indexed with $k$\end{tabular}\\ \hline
$\rho(\y_j^k)$ & Actual execution time of job $j$ when schedule $\y_j^k$ is used\\ \hline
$\hat{\rho}(\y_j^k)$ & \begin{tabular}[c]{@{}l@{}l@{}}Estimated execution time of job $j$ when schedule $\y_j^k$ is used\end{tabular}\\ \hline
$\mathcal{G}(\y_j^k)$ & \begin{tabular}[c]{@{}l@{}l@{}}Set of GPUs allocated for job $j$ when schedule $\y_j^k$ is used\end{tabular}\\ \hline
$x_j^k$ & \begin{tabular}[c]{@{}l@{}l@{}}Indicate whether job $j$ follows schedule $\y_j^k$ or not\end{tabular}\\ \hline
$W_{jg}^k$ & \begin{tabular}[c]{@{}l@{}l@{}}Execution time added to GPU $g$ by job $j$ if job j follows $\y_j^k$\end{tabular}\\ \hline
$U_s^g$ & \begin{tabular}[c]{@{}l@{}l@{}} The accumulative execution time of worker $g$ on server $s$\end{tabular}\\ \hline
%
\end{tabular}
}
\vspace{-.1in}
\end{table}


\section{Solution Approach} \label{sec:alg}
As mentioned in Section~\ref{sec:model_formulation}, a key challenge to solve Problem~RAR-DDLS is that, due to the mixed covering- and packing-type constraints, the number of job scheduling combinations grows exponentially as the number of servers/jobs increases.
Thus, it is computationally prohibitive to enumerate all possible combinations before the scheduler decides when to start and which GPU(s) should be allocated to achieve the optimal scheduling.
Exacerbating the problem is the fact that communication contention renders a mixed-integer bilinear structure in (\ref{eqn:p_j}), making it intractable to express $p_{j}[t]$ in closed-form.
Due to these challenges, it is difficult to directly solve Problem~RAR-DDLS based on its original formulation.
To overcome this challenge, we propose the following ``indirect'' approach to solve Problem~RAR-DDLS.

\smallskip
{\bf 1) Basic Idea:}
First, we note that, although not in closed-form expressions, the per-iteration time $\tau_j[t]$ for each job can be computed in polynomial time according to (\ref{eqn:p_j})-(\ref{contention_model}) once a schedule (i.e., $\y[t] = \{y_{js}[t], \forall j, s\}$) is given.
Specifically, we note that the per-iteration time $\tau_j[t]$ is determined by $B_j(\y[t])$ and $\gamma_j(\y_j[t])$.
Moreover, $f(\boldsymbol{\alpha}, k_j[t])$ increases as $k_j[t]$ gets larger, and $\gamma_j(\y_j[t])$ increases as $\sum_s\mathds{1}\{y_{js}[t]>0\}$ grows.
Thus, the range of $\tau_j[t]$ can be estimated.
The largest number of $k_j[t]$ is $\max_s O_s$, i.e., the worst case would be each job places one of its workers into the server with the biggest capacity and they all compete for the bandwidth.
Thus, we can have $B_j(\y[t]) \in [b^e/f(\boldsymbol{\alpha}, \max_s O_s), b^i]$.
 In addition, we have $\sum_s\mathds{1}\{y_{js}[t]>0\}\in[1, G_j]$.
 Then by plugging $B_j(\y[t])$ and $\sum_s\mathds{1}\{y_{js}[t]>0\}$ with their lower and upper bounds in Eqn.~(\ref{contention_model}), respectively, we can attain the lower and upper bounds.

The above insight suggests that we can solve Problem~RAR-DDLS via the following {\em search-based approach} to circumvent the structural difficulty in (\ref{eqn:p_j})-(\ref{contention_model}) .
We can first search for a schedule $\y$, then $\tau_j[t],\forall t$ can be efficiently evaluated to estimate the makespan.
Then, we repeat the process until we find a ``good enough'' schedule.
Therefore, we can have the algorithmic framework as shown in Fig.~\ref{fig:idea} to obtain an approximate makespan if the search space is given.
Clearly, the search space of $\y$ remains huge and difficult to sample.
Nonetheless, in what follows, we show that Problem~RAR-DDLS can be reformulated to facilitate this search-based approach.

\begin{figure}[t!]
    \centering
    \includegraphics[width=0.45\textwidth]{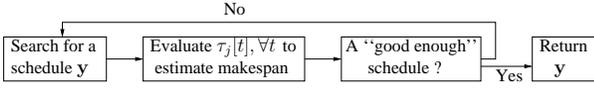}
    \caption{Basic idea for solving Problem~RAR-DDLS.}\label{fig:idea}
    \vspace{-.1in}
\end{figure}

\smallskip
{\bf 2) Problem Reformulation:}
In order to enable the search of a schedule, we first reformulate Problem~RAR-DDLS by introducing following notations.
We let $\mathcal{N} = \{1, \hdots, N\}$ be the set of all GPUs in the cluster.
Let $\mathcal{Y}=\{\y^1,\hdots,\y^{|\mathcal{Y}|}\}$ be the set of feasible scheduling schemes for the jobs to be scheduled, where $\y^k=\{\y_1^k,\hdots,\y_J^k\}$ and $\y_j^k=\{y^k_{js}[t], \forall s\in\mathcal{S}, t\in\mathcal{T}\}\in\mathbb{Z}_+^{S\times T}$.
Note that, with a slight abuse of notation, we use $y^k_{js}[t]$ here as a constant (not a variable) to denote the number of workers allocated for job $j$ on server $s$ in time-slot $t$ if  schedule $\y^k$ is used.
We also use $\rho_j(\y^k)$ to denote the execution time of job $j$ if schedule $\y^k$ is used.
Also, we denote the starting time of job $j$ under schedule $\y^k$ as $a_j(\y^k)\triangleq \arg\min\{t | y^k_{js}[t] > 0, \exists s\}$.
Let $x_j^k \in \{0,1\}$ be the binary variable to indicate whether job $j$ follows schedule $\y^k$ ($x_j^k=1$) or not ($x_j^k=0$).
Then Problem~RAR-DDLS can be reformulated as the following integer linear program (ILP):
\begin{align}
\label{problem:makespan}\min_{x_j^k,\forall j,k} \max_j& \ x_j^k\big(a_j(\y^k)+\rho_j(\y^k)\big)\\
\text{subject to. } &\sum_{k\in\{1,\hdots,|\mathcal{Y}|\}}x_j^k = 1,\quad \forall j\in\mathcal{J},\label{eqn:all_placed}\\
&x_j^k = x_{j'}^k, \quad \forall j, j'\in\mathcal{J}, k\in\{1,\hdots,|\mathcal{Y}|\},\label{eqn:placement}\\
&x_j^k\in\{0,1\}, \quad\forall j\in\mathcal{J}, k\in\{1,\hdots,|\mathcal{Y}|\}.\label{eqn:binary}
\end{align}
Constraint~(\ref{eqn:all_placed}) ensures that exactly one schedule is chosen.
Constraint~(\ref{eqn:placement}) ensures that all jobs use the same schedule $\y^k$.
We note that, although Problem~(\ref{problem:makespan}) has a simpler structure compared to Problem~RAR-DDLS, it hides the complexity in the dimensionality of the exponential search space $\mathcal{Y}$, which is intractable to explore.
However, based on this reformulated problem, we will show next that it is possible for one to identify a ``good enough'' schedule such that the makespan can be upper bounded.

Unfortunately, Problem~(\ref{problem:makespan}) remains an NP-hard problem.
We state this formally in Theorem~\ref{thm:np}, which can be proved based on the reduction to the vertex coloring problem (VCP). 

\begin{thm}\label{thm:np}
Let $n_g=\max_j G_j$.
Solving Problem~\eqref{problem:makespan} to within an $O(\frac{\log n_g}{2^{\sqrt{\log\log n_g}}})$-approximate ratio is NP-hard even when the exact processing time of each job is available.
\end{thm}
\begin{proof}
Here, we consider the special case with all jobs having a unit processing time ($\rho(\y^k)=1$).
We first show that VCP can be reduced to the job scheduling problem in Problem~(\ref{problem:lp}) in polynomial time.
Given an instance $I$ of VCP, i.e., given a graph $G=(V, E)$ with bounded degree $n_g$, we can construct our job scheduling problem as follows: i) For each node $v\in V$, we create a job $j_v\in\mathcal{J}$, where it has only one schedule $y_{j_v} = \emptyset$;
ii) For each edge $(u, v)\in E$,  we add a worker $w_{u,v}$ to $\mathcal{S}$. Also, update the scheduling as $y_{j_u}=y_{j_u} \cup \{w_{u,v}\}$ and $y_{j_v}=y_{j_v} \cup \{w_{u,v}\}$.
If the graph $G's$ maximum degree is no greater than $n_g$, then the maximum number of workers that can be allocated to each job is also $n_g$.

With this reduction, we next show the solution of VCP can be translated to the solution of Problem~(\ref{problem:lp}), and vice verse.
First, recall that each job has one unit processing time.
Thus, all jobs should be executed inside a unit time interval ($[0, 1), [1, 2), \hdots$).
If we have the solution to VCP, then we can schedule jobs with the same color in the same interval.
Also, if we have the solution to the job scheduling problem, then the jobs in the same interval can be marked as the same color.
Hence, finding an optimal solution of Problem~(\ref{problem:lp}) is equivalent to find an optimal solution of VCP.

Similarly, given an instance of job scheduling, we can construct an instance of VCP, where the makespan equals to the number of colors.
Therefore, if we have an $\alpha$-approximate solution to the job scheduling problem, we have an $\alpha$-approximate solution to VCP.
However, 
with the graph of degree at most $n_g$, it is known that coloring a $2^{\sqrt{\log\log n_g}}$-colorable graph with $O(\log n_g)$ colors is NP-Hard.
This completes the proof.
\end{proof}
The hardness result in Theorem~\ref{thm:np} suggests that solving Problem~(\ref{problem:makespan}) necessitates the design of approximation algorithms, which is our goal in algorithm development next.

\smallskip
{\bf 3) Identify a Scheduling with Bounded Makespan:}
We let $\mathcal{G}_j(\y^k)$ be the set of GPUs allocated for job $j$ when schedule $\y^k$ is used.
We use $W_{jg}^k=x_j^k\rho_j(\y^k)$ to denote the execution time added to GPU $g$ by job $j$ if job $j$ follows schedule $\y^k$. 
Since each job $j$ only chooses one schedule, the total execution time of GPU $g$ can be computed as $W_g=\sum_j\sum_kW_{jg}^k$.
However, due to communication contention, the exact processing time $\rho_j(\y^k)$ is hard to evaluate in computing $W_{jg}^{k}$. 
Fortunately, the estimated processing time $\hat{\rho}_j(\y^k)$ can be bounded as $\hat{\rho}_j(\y^k) \in [l\rho_j(\y^k), u\rho_j(\y^k)]$ for some $l\leq 1$ and $u\geq 1$, since $\tau_j[t]$ is bounded. 
Here, we use $\frac{\hat{\rho}_j(\y^k)}{u}\leq \rho_j(\y^k)$ to replace $\rho_j(\y^k)$ when computing $W_{jg}^{k}$.
Consider a search algorithm $\pi$ that solves the following ILP to choose one schedule from $\mathcal{Y}$:
\begin{align}
\label{problem:lp}&\min \  -1\\
\text{subject to. }
&\hat{W}_{jg}^k = x_j^k\frac{\hat{\rho}_j(\y^k)}{u}, \forall j\!\in\!\mathcal{J}, k\!\in\!\{1,\hdots,|\mathcal{Y}|\}, g\!\in\!\mathcal{N},\label{eqn:wk}\\
&\sum_{j\in\mathcal{J}}\sum_{k\in\{1,\hdots,|\mathcal{Y}|\}}\hat{W}_{jg}^k\leq \theta_u, \quad \forall g\in\mathcal{N},\label{eqn:gpu_lp}\\
\nonumber&\text{Constraints}~(\ref{eqn:all_placed})-(\ref{eqn:binary}).
\end{align}
Note that Problem~(\ref{problem:lp}) has no objective function to be optimized since we are only interested in whether a feasible solution no greater than a given maximum execution time limit $\theta_u$ exists ($\theta_u$ depends on parameter $u$).
Constraints~(\ref{eqn:wk})--(\ref{eqn:gpu_lp}) ensure that no GPU's execution time would exceed $\theta_u$.
Let $W^{\pi}_{\max}=\max_{g\in\mathcal{N}}W_g$ be the maximum execution time of all GPUs returned by algorithm $\pi$.
Due to the use of estimated $\frac{\hat{\rho}_j(\y^k)}{u}$, the solution of $\pi$ finds a lower bound of $W_{\max}^{\pi}$, which is also a lower bound of the makespan under $\pi$ (due to potential idling resulted from synchronization barrier). 

Note that for any feasible scheduling with the upper bound $\theta_u$ for Problem~(\ref{problem:lp}), we can find a corresponding feasible solution for Problem~(\ref{problem:makespan}) by setting $x_j^k=1$ if job $j$ follows schedule $\y^k$; otherwise, set $x_j^k = 0$.
Thus, the challenge of solving Problem~(\ref{problem:makespan}) becomes finding a tightest execution time limit $\theta_u$ for Problem~(\ref{problem:lp}), which is relatively easy since there is no need to explore the exponential search space of schedules $\mathcal{Y}$.

It is insightful to understand the choice of $\theta_u$ in Problem~\eqref{problem:lp}.
On one hand, if $\theta_u$ is too small, Problem~(\ref{problem:lp}) could be infeasible, and no scheduling for Problem~(\ref{problem:makespan}) can be found.
On the other hand, when $\theta_u$ is too large, then all schedulings can be considered, and the gap between the optimal maximum execution time and the optimal makespan can be large, thus no meaningful lower bound of $W^\pi_{\max}$ can be found.
Fortunately, since determining an appropriate $\theta_{u}$ is a {\em univariate} search, we can simply use the bisection method to efficiently find the minimum $\theta_u$ feasible to Problem~(\ref{problem:lp}).

\smallskip
{\bf 4) Algorithm Description:}
We next present our scheduling algorithm based on bisection to search $\theta_u$ and the smallest job first strategy to solve Problem~\eqref{problem:lp} for a given $\theta_u$. 
Note that if a job's ring of workers is scheduled over a large number of servers, it may potentially worsen communication contention with concurrent jobs and its communication overhead could be large.
Therefore, to control the number of active servers, we use a threshold parameter $\kappa \in[1, n_g]$  to control the number of maximum servers for scheduling jobs.
We summarize our scheduling approach in Algorithm~\ref{alg:place}.
The intuition behind Algorithm~\ref{alg:place} is that: 
1) When the job is small (i.e., $G_j \leq \kappa$), we prefer to pack the job into servers whose GPUs are already occupied by some other jobs rather than opening new server(s) to host its workers. 
Since the job is small, the induced contention is mild by using the shared servers. 
Further, by packing its workers to these servers, we can avoid fragmentation introduced by a small job and save space for larger jobs that will be scheduled next.
2) If $G_j > \kappa$, we prefer to allocate the job's workers to new server(s).
This is because shared servers may only have limited available GPU(s), and in order to gang-schedule a large job, a large number of shared servers may be used, which leads to a high communication overhead.

 \begin{algorithm}[t!]
\SetAlgoLined
\textbf{Input:} $\mathcal{J}$, $U_s^g$, $\hat{\rho}_j(\y^k)$, $u$, $\lambda_j$\;
\textbf{Initialization:} Let $U_s^g\leftarrow0,\forall g, s$\;

Sort jobs by $G_j$ in non-decreasing order\label{line:sort-J}, and denote as $\mathcal{J}^{s}$\;
$m \leftarrow T$, $\y \leftarrow \emptyset$, $left\leftarrow 1$, $right\leftarrow T$\;
\While {$left <= right$\label{line:bs}}{
	$\theta_u\leftarrow (left+right)/2$, $m_{\theta}\leftarrow T$, $\y_\theta\leftarrow \emptyset$\;
	\For {$\kappa = 1,2,\hdots,\max_jG_j$\label{line:kappa}}{
		$\y_{\theta}^k\leftarrow\emptyset$, m$_{\theta}^k\leftarrow -1$\;
		\For {$j = 1,2,\hdots,|\mathcal{J}^s|$\label{line:iterate}}{
			\If{$G_j\leq \kappa$\label{line:<=}}{
			 Return $\y_j$, $T_j$ using Algorithm~\ref{alg:ff}\label{line:ff}\;
		}
			\Else{
				Return $\y_j$, $T_j$ using Algorithm~\ref{alg:ls}\label{line:ls}\;
		}
			\If{$\y_j==\emptyset$\label{line:infeasible}}{
				break\;
			}
			$\y_{\theta}^k\leftarrow \y_{\theta}^k\cup \{\y_j\}$, $m_{\theta}^k\leftarrow \max\{\text{m}_{\theta}^k, T_j\}\label{line:k-update}$\;
		}
		
		\If{$m_{\theta}^k < m_\theta$\label{line:local-compare}}{
			$m_\theta \leftarrow m_{\theta}^k$, $\y_\theta\leftarrow\y_\theta^k$\label{line:theta-update}\;
		}
	}
	\If{$m_\theta < m$\label{line:global-compare}}{
		$m\leftarrow m_\theta$, $\y\leftarrow\y_\theta$\label{line:global-update}\;
		$right\leftarrow \theta_u - 1$\label{line:right-move}\;
	}
	\Else{
		$left\leftarrow \theta_u + 1$\label{line:left-move}\;
	}
}
\Return m, $\y$\;

\caption{\underline{S}mallest \underline{J}ob \underline{F}irst with \underline{B}alanced \underline{C}ontention and \underline{O}verhead (SJF-BCO).}
 \label{alg:place}
\end{algorithm} 

\begin{algorithm}[t!]
\SetAlgoLined
\textbf{Input:} A given job $j$, $\mathcal{S}$, $U_s^g$, $\hat{\rho}_j(\y^k)$, $u$, $\theta_u$\;
	$\mathcal{G}^{\theta_u}_{idle}\leftarrow$ available GPUs with execution time not exceed $\theta_u$\label{line:add}\;
	     \If{$|\mathcal{G}^{\theta_u}_{idle}|\geq G_j\label{line:enough}$}{
	         Pick top-$G_j$ workers with least $U_s^g$ from $\mathcal{G}^{\theta_u}_{idle}$ as $\y_j$\label{line:LS}\;
		$T_j\leftarrow\arg\max_t\{y_{js}[t] > 0 | y_{js}[t]\in\y_j, \forall s, t\}$\label{line:t}\;
	         $U_s^g \leftarrow U_s^g + \hat{\rho}_j(\y^k)/u, \forall (g,s)\in\y_j\label{line:update-U}$\;
	         \Return $\y_j, T_j$\;
	     }
	   \If{there exists running jobs\label{line:if-ff}}{
	   	Waiting for some job to exit and then goes to Line~\ref{line:add}\label{line:repeat-ff}\;
	   }
	   \Return $\emptyset, T$\;

\caption{Fragment-Aware First Fit Packing (FA-FFP).}
 \label{alg:ff}
\end{algorithm} 

 \begin{algorithm}[t!]
\SetAlgoLined
\textbf{Input:} A given job $j$, $U_s^g$, $\hat{\rho}_j(\y^k)$, $u$, $\lambda_j$\;
	    Sort the server set $\mathcal{S}$ by $\sum_gU_s^g/O_s$ in non-decreasing order, and choose the top $m$-servers s.t. $\sum_{s=1}^{m}O_s \geq \lambda_jG_j$, and denote the selected server set as $\mathcal{S}_{selected}$\label{line:sort-S}\;
	    $\mathcal{G}^{\theta_u}_{idle}\leftarrow \emptyset$\;
	    \For{$s\in\mathcal{S}_{selected}\label{line:for}$}{
	        Sort GPUs whose $U_s^g+\hat{\rho}_j(\y^k)/u\leq\theta_u$ by execution time in non-decreasing order, then append them to $\mathcal{G}^{\theta_u}_{idle}$\label{line:sort-U}\;
	     }
	        	\If{$|\mathcal{G}^{\theta_u}_{idle}| \geq G_j\label{line:gpus}$}{
	        Pick top-$G_j$ workers with least $U_s^g$ as $\y_j$\label{line:pick}\;
	        $T_j\leftarrow\arg\max_t\{y_{js}[t] > 0 | y_{js}[t]\in\y_j, \forall s, t\}$\;
	        $U_s^g\leftarrow U_s^g + \hat{\rho}_j(\y^k)/u, \forall (g,s) \in \y_j$\label{line:update-U-2}\;
	        \Return $\y_j, T_j$\label{line:return}\;
	    	}
	      \If{there are running jobs\label{line:if-ls}}{
	      	Waiting for some job to exit and then goes to Line~\ref{line:sort-S}\label{line:repeat-ls}\;
	      }
	      \Return $\emptyset, T$\label{line:return-ls}\;
\caption{Least Busy Server-GPU First (LBSGF).}
 \label{alg:ls}
\end{algorithm} 

In Algorithm~\ref{alg:place}, $U_s^g$ denotes the accumulative execution time of worker $g$ on server $s$.
We first sort jobs in non-decreasing order of their sizes $G_j$ in Line~\ref{line:sort-J}.
We search $\theta_u$ using the bisection method in the range $[1, T]$, and use the pair $(\theta_u,\kappa)$ to perform scheduling (Lines~\ref{line:bs}-\ref{line:kappa}).
We then iterate through each job (Line~\ref{line:iterate}).
If its size is not greater than the threshold $\kappa$ (Line~\ref{line:<=}), Algorithm~\ref{alg:ff} will be used to do the scheduling (Line~\ref{line:ff}); otherwise, Algorithm~\ref{alg:ls} will be called (Line~\ref{line:ls}).
If no feasible scheduling of job $j$ is returned, then we quit the current loop and update $\kappa$ (Line~\ref{line:infeasible}); otherwise, we will update the scheduling and makespan given the current $(\theta_u, \kappa)$ (Line~\ref{line:k-update}).
Upon finishing scheduling all jobs, we will update the schedule and makespan for the given $\theta_u$ if it has a smaller makespan (Lines~\ref{line:local-compare}-\ref{line:theta-update}).
After exhausting all values of $\kappa$ for a given $\theta_u$, we will update the global makespan and the schedule if the current input $\theta_u$ has a better performance (Lines~\ref{line:global-compare}-\ref{line:global-update}).
Also, it indicates that we can further decrease the value of $\theta_u$ to find a potentially better schedule.
Thus, we search for the left half space by moving the right pointer (Line~\ref{line:right-move}); otherwise, we should increase the value of $\theta_u$ by moving the left pointer (Line~\ref{line:left-move}).
By scheduling workers as described in Algorithm~\ref{alg:place}, no worker's execution time will exceed the given limit $\theta_u$.
We denote the tightest execution time limit returned as $\tilde{\theta}_u$.

Algorithm~\ref{alg:ff} is based on the idea of ``fragment-aware first fit packing,'' where we first add all available GPUs whose $U_s^g+\hat{\rho}_j(\y^k)/u \leq \theta_u$ (Line~\ref{line:add}).
If there are enough available GPUs to schedule for job $j$'s workers (Line~\ref{line:enough}), we choose top-$G_j$ GPUs with least execution time first (Line~\ref{line:LS}).
We then evaluate the completion time of job $j$ (Line~\ref{line:t}) and update the corresponding GPUs' execution time (Line~\ref{line:update-U}); otherwise, we wait for some job to finish (Lines~\ref{line:if-ff}-\ref{line:repeat-ff}).

Algorithm~\ref{alg:ls} is based on the idea of ``least busy server-GPU first,'' where we sort the servers by its GPU's average accumulative execution time (Line~\ref{line:sort-S}) and add the available GPUs whose execution time does not exceed $\theta_u$ in a non-decreasing order (Lines~\ref{line:for}-\ref{line:sort-U}).
Here, we introduce $\lambda_j\geq 1$ as a tuning parameter.
The smaller the $\lambda_j$ is, the fewer number of servers can be used.
If enough idle workers can be found, we schedule the job, evaluate its completion time, update the execution time of the chosen GPUs, and return the schedule (Lines~\ref{line:gpus}-\ref{line:return}); otherwise, we wait for some job to finish (Lines~\ref{line:if-ls}-\ref{line:repeat-ls}).
If there is no running job left, then return schedule $\emptyset$ and timespan $T$ (as makespan) to indicate the scheduling is infeasible (Line~\ref{line:return-ls}).

\section{Performance Analysis}\label{sec:Performance}
In this section, we analyze the theoretical performance of SJF-BCO.
Specifically, we will establish the approximation ratio guarantee of our proposed SJF-BCO algorithm as follows:

\begin{list}{\labelitemi}{\leftmargin=1.2em \itemindent=-0.0em \itemsep=.2em}
\item[1)] We first show in Lemma~\ref{lem:execution time} that the maximum execution time (i.e., $\hat{W}_{\max}^{\mathrm{Alg1}}$) returned by our algorithm is equal to $\tilde{\theta}_u$.

\item[2)] We then prove that the makespan is $O(\hat{W}_{\max}^{\mathrm{Alg1}})$ in Lemma~\ref{lem:makespan}. 

\item[3)]  We further show that the gap between $\tilde{\theta}_u$ and the tightest execution time limit $\theta_u^*$ returned by some offline optimal algorithm in the right-hand-side (RHS) of \eqref{eqn:gpu_lp} is bounded in Lemma~\ref{lem:wk_gap}. 
\end{list}
Finally, by putting all these lemmas together, we arrive at the approximation ratio result stated in Theorem~\ref{thm:ratio}.

\begin{lem}[Maximum Execution Time Upperbound]\label{lem:execution time}
Algorithm~\ref{alg:place} produces a schedule with the maximum execution time $\hat{W}_{\max}^{\mathrm{Alg1}} =\tilde{\theta}_u$.
\end{lem}

\begin{proof}
Note that in Algorithm~\ref{alg:place}, we can obtain a schedule such that the execution time of every worker will not exceed $\tilde{\theta}_u$,
i.e., $\sum_j\sum_k x_j^k\frac{\hat{\rho}_j(\y^k)}{u}\leq \tilde{\theta}_u$, $\forall g$ (cf. Line~2 in Algorithm~\ref{alg:ff} and Line~5 in Algorithm~\ref{alg:ls}).
Note that $\tilde{\theta}_u$ is the tightest value found by Alg.~\ref{alg:place} since we will keep decreasing its value in the RHS of \eqref{eqn:gpu_lp} until it becomes equal to $\hat{W}_{\max}^{\mathrm{Alg1}}$ in the LHS of \eqref{eqn:gpu_lp}.
It then follows that:
\begin{align*}
\hat{W}_{\max}^{\mathrm{Alg 1}} = \max_{g\in\mathcal{N}}\sum_{j\in\mathcal{J}}\sum_{k\in\{1,\hdots,|\mathcal{Y}|\}}x_j^k\frac{\hat{\rho}_j(\y^k)}{u} = \tilde{\theta}_u.
\end{align*}
Thus, we can have the maximum execution time $\hat{W}_{\max}^{\mathrm{Alg1}}$ is equal to $\tilde{\theta}_u$, and the proof is complete.
\end{proof}


\begin{lem}[Makespan Upperbound]\label{lem:makespan}
Algorithm~\ref{alg:place} achieves a makespan at most $n_g\hat{W}^{\mathrm{Alg1}}_{\max}$, where $n_g$ is defined as in Theorem~\ref{thm:np}. 
\end{lem}

\begin{proof}
To bound the makespan, we need to attain upper bounds of the total {\it busy} and {\it idle} time for each worker. 
Recall that due to the synchronous gang scheduling for training, the worker may wait for other workers to finish executing other jobs before it could start processing the current job, which may result in idling.
First, we can have the total busy time $T_g^{\mathrm{busy}}\leq \hat{W}^{\mathrm{Alg1}}_{\max}\overset{\mathrm{Lem.}~\ref{lem:execution time}}{=}\tilde{\theta}_u$. 
Next, we work on bounding the total idle time $T_g^{\mathrm{idle}}$.

For any worker $g\in\mathcal{N}$, we use $g_j$ to denote the last job $j$ on $g$.
Suppose job $j$ follows schedule $\y^k$.
At any time slot $t$ before worker $g$ processes job $j$, there are two cases: i) worker $g$ is occupied by other jobs (i.e., $g$ is busy); ii) worker $g$ is idle, but at least one worker $g'\in\mathcal{G}_j(\y^k)$ is busy with executing other jobs.
Since we consider the gang-scheduling discipline, 
the job cannot be delayed if there is a sufficient number of GPUs available as requested.
Thus we have:
\begin{align*}
T_g^{\mathrm{idle}} \overset{\mathrm{(a)}}{\leq}\!\!\!\!\sum_{g'\in\mathcal{G}_j(\y^k)|g'\neq g}\!\!\!\!T_{g'}^{\mathrm{busy}}\overset{}{\leq}\sum_{g'\in\mathcal{G}_j(\y^k)|g'\neq g}\!\!\!\!\hat{W}^{\mathrm{Alg 1}}_{\max}\overset{\mathrm{(b)}}{\leq}(G_j\!-\!1)\hat{W}^{\mathrm{Alg 1}}_{\max},
\end{align*}
where (a) follows from the fact that in any time slot $t$ that worker $g$ is idle (case ii), we must be able to find at least one busy worker $g'\in\mathcal{G}_j(\y^k)$.
To calculate the idle time of worker $g$, we can calculate the busy time of worker(s) $g'\in\mathcal{G}_j(\y^k)$ instead, and the limit of each worker's busy time is $\hat{W}_{\max}^{\mathrm{Alg1}}$.
Also, (b) follows from the fact that at most $G_j-1$ number of GPUs (except worker $g$)  are busy. 
Then, we can upper bound the makespan $T^{\mathrm{total}}$ as:
\begin{align*}
T^{\mathrm{total}} &= \max_{g\in\mathcal{N}}(T_g^{\mathrm{busy}} + T_g^{\mathrm{idle}})\leq  \max_{j\in\mathcal{J}}\bigg(\hat{W}_{\max}^{\mathrm{Alg 1}} + (G_j-1)\hat{W}_{\max}^{\mathrm{Alg1}}\bigg)\\
&=\max_{j\in\mathcal{J}}G_j\hat{W}_{\max}^{\mathrm{Alg 1}}= n_g\hat{W}_{\max}^{\mathrm{Alg 1}},
\end{align*}
and the proof is complete.
\end{proof}

Next, we characterize the gap between the maximum execution time limit $\tilde{\theta}_u$ and the optimal execution time $\theta^*_u$ in the RHS of \eqref{eqn:gpu_lp}.

\begin{lem}\label{lem:wk_gap}
The maximum execution time $\tilde{\theta}_u$ returned by Algorithm~\ref{alg:place} satisfies $\tilde{\theta}_u\leq \varphi\frac{u}{l}\cdot\theta^*_u$, where $\varphi=\max_j\frac{\rho_j(\y^{k_1})}{\rho_j(\y^{k_2})}, \forall k_1, k_2$.
\end{lem}

\begin{proof}
Let $k^*$ and $\tilde{k}$ be the schedule indices chosen by solving Problem~(\ref{problem:lp}) optimally and Algorithm~\ref{alg:place}, respectively.
Let $\mathcal{G}(\y^{k})$ be the set of selected GPUs if schedule $\y^k$ is used.
We have
\begin{align*}
\tilde{\theta}_u &\overset{\mathrm{Lem.~\ref{lem:execution time}}}{=} \max_{g\in\mathcal{G}(\y^{\tilde{k}})} \sum_{j\in\mathcal{J}}\frac{\hat{\rho}_j(\y^{\tilde{k}})}{u}\overset{\mathrm{(a)}}{\leq} \max_{g\in\mathcal{G}(\y^{\tilde{k}})} \sum_{j\in\mathcal{J}}\frac{\varphi\frac{u}{l}\hat{\rho}_j(\y^{k^*})}{u}\\
&\overset{\mathrm{(b)}}{\leq} \max_{g\in\mathcal{G}(\y^{k^*})} \sum_{j\in\mathcal{J}}\frac{\varphi\frac{u}{l}\hat{\rho}_j(\y^{k^*})}{u} \overset{\mathrm{Eq.}~\eqref{eqn:gpu_lp}}{\leq} \varphi\frac{u}{l}\theta^*_u.
\end{align*}
To see why (a) holds, recall that for any schedule $\y^k$, we have $\hat{\rho}_j(\y^k)\in[l\rho_j(\y^k), u\rho_j(\y^k)]$.
Then, for any two different schedules $\y^{k_1}$ and $\y^{k_2}$, we can calculate the worst-case ratio as $\frac{\hat{\rho}_j(\y^{k_1})}{\hat{\rho}_j(\y^{k_2})}\leq \frac{u\rho_j(\y^{k_1})}{l\rho_j(\y^{k_2})}\leq\varphi\frac{u}{l}$.
The inequality in (b) can be established as follows.
First, note that $\tilde{k}$ is chosen using the {\it least execution time first} scheduling strategy in Algorithm~\ref{alg:ff} (Line~\ref{line:LS}).
Then, we have $\max_{g\in\mathcal{G}(\y^{\tilde{k}})} \sum_{j\in\mathcal{J}}\frac{\hat{\rho}_j(\y^{\tilde{k}})}{u}\leq \max_{g\in\mathcal{G}(\y^{k})} \sum_{j\in\mathcal{J}}\frac{\hat{\rho}_j(\y^{k})}{u}$, $\forall k$, which can be proved by contradiction as follows.
Suppose there exists $g\in\mathcal{G}(\y^{k})\setminus \mathcal{G}(\y^{\tilde{k}})$ such that $\sum_{j\in\mathcal{J}}\frac{\hat{\rho}_j(\y^{k})}{u}\leq \sum_{j\in\mathcal{J}}\frac{\hat{\rho}_j(\y^{\tilde{k}})}{u}$, $\forall g'\in\mathcal{G}(\y^{\tilde{k}})$.
However, we know that $\tilde{k}$ chooses the GPUs with the least execution time first, i.e., $g$ should be in $\mathcal{G}(\y^{\tilde{k}})$, which contradicts our assumption.
This completes the proof.
\end{proof}

Finally, by putting everything together, we have the following approximation ratio for our proposed approach:
\begin{thm}[Approximation Ratio]\label{thm:ratio}
Alg.~\ref{alg:place} is $n_g\varphi\frac{u}{l}$-approximate.
\end{thm} 
\begin{proof}
We use $T^*$ to denote the optimal makespan that produced by some offline optimal algorithm.
It then follows that
\begin{align*}
T^{total}&\overset{\mathrm{Lem.}\ref{lem:makespan}}{\leq}n_g\hat{W}^{\mathrm{Alg1}}_{\max}\overset{\mathrm{Lem.}\ref{lem:execution time}}{=}n_g\tilde{\theta}_u\overset{\mathrm{Lem.}\ref{lem:wk_gap}}{\leq}n_g\varphi\frac{u}{l}\theta^*_u\overset{\mathrm{(a)}}{\leq}n_g\varphi\frac{u}{l}T^{*},
\end{align*}
where (a) is due to Problem~(\ref{problem:lp}) estimates the processing time as $\frac{\hat{\rho}_j(\y^k)}{u}$ without considering potential idling (caused by synchronization barrier), which implies $\theta^*_u\!\leq\!T^{*}$.
This completes the proof.
\end{proof}

\begin{remark}{\em
Note that the result in Theorem~\ref{thm:ratio} does not depend explicitly on the parameter $\kappa$ in SJF-BCO.
This is because Theorem~\ref{thm:ratio} is only a worst-case upper bound that depends on $\tilde{\theta}_u$, which in turn depends on $\kappa$.
Hence, $\kappa$ is implicitly captured in Theorem~\ref{thm:ratio}. 
}
\end{remark}

\begin{thm}[Polynomial Running Time]\label{thm:time}
Time complexity of SJF-BCO is $O(n_g|\mathcal{J}|N\log N\log T)$, where $n_g$ is defined as in Thm.~\ref{thm:np}.
\end{thm}
\begin{proof}
The sorting operation plays a dominant role in the total running time in Algorithm~\ref{alg:place}.
For each job $j$, if $G_j\leq\kappa$, we need to sort all GPUs in the cluster, which takes $O(N\log N)$ time in order to choose top-$G_j$ workers with least execution time first in Algorithm~\ref{alg:ff} (Line~\ref{line:LS}).
Otherwise, we only need to sort servers, which takes $O(S\log S)$ time in order to choose top-$m$ servers as in Algorithm~\ref{alg:ls} (Line~\ref{line:sort-S}).
Thus, it takes $O(N\log N)$ time to schedule each job since $N > S$.
Then, for all the jobs to be scheduled given ($\theta_u, \kappa$), it has $O(|\mathcal{J}|N\log N)$ time complexity.
Recall that we use bisection to search $\theta_u$, where each iteration contains an inner loop indexed by $\kappa\!\in\![1, n_g]$.
This implies a total of $n_g\log T$ trials.
Thus, the overall time complexity is $O(n_g|\mathcal{J}|N\log N\log T)$.
\end{proof}


\begin{figure*}[t!]
       \begin{minipage}[t]{0.24\linewidth}
        \centering
        \includegraphics[width=1.05\textwidth]{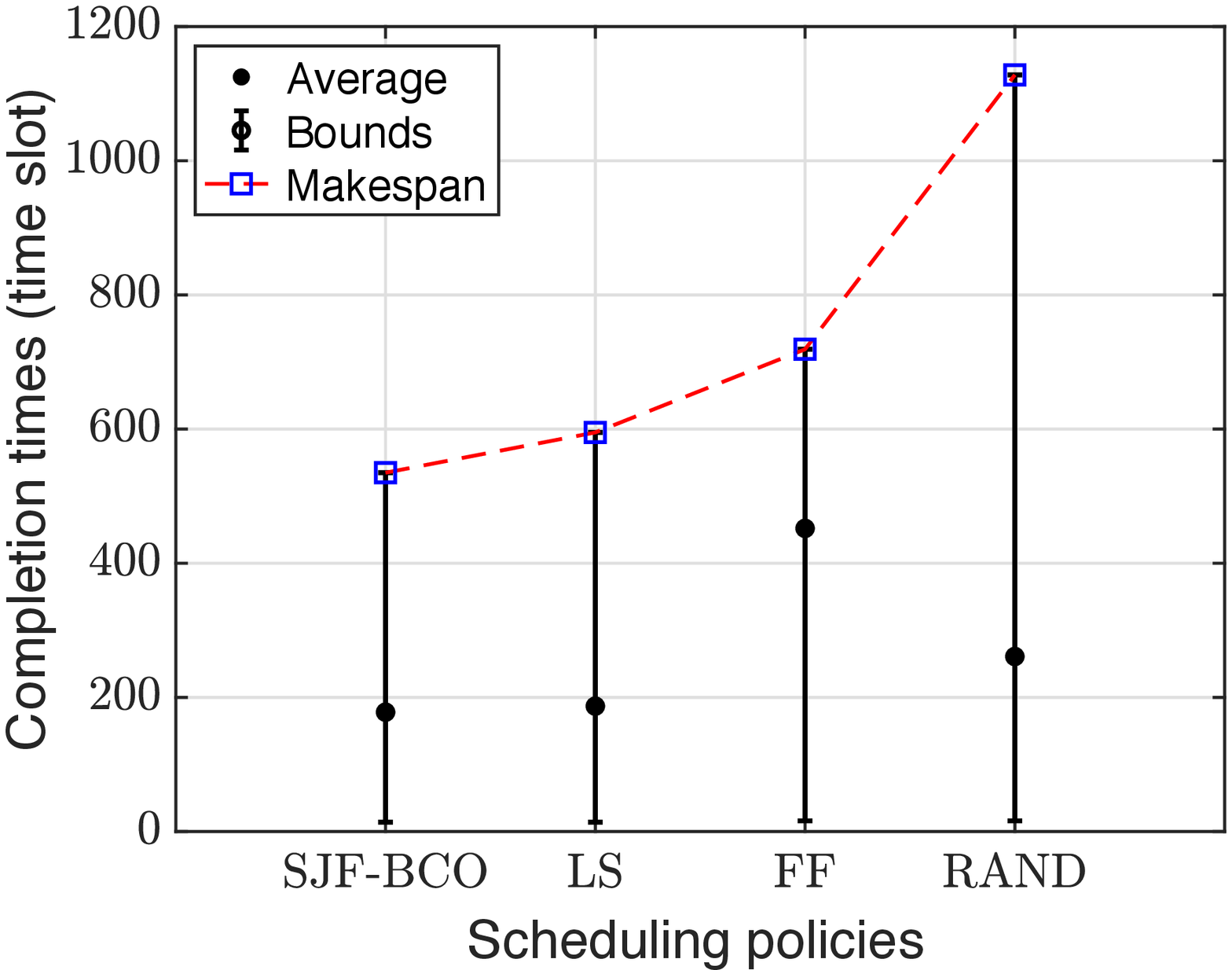}
        \caption{Makespan comparison under different policies.}\label{fig:baseline} 
    \end{minipage}%
    \hspace{0.006\linewidth}
    \begin{minipage}[t]{0.24\linewidth}
        \centering
        \includegraphics[width=1.04\textwidth]{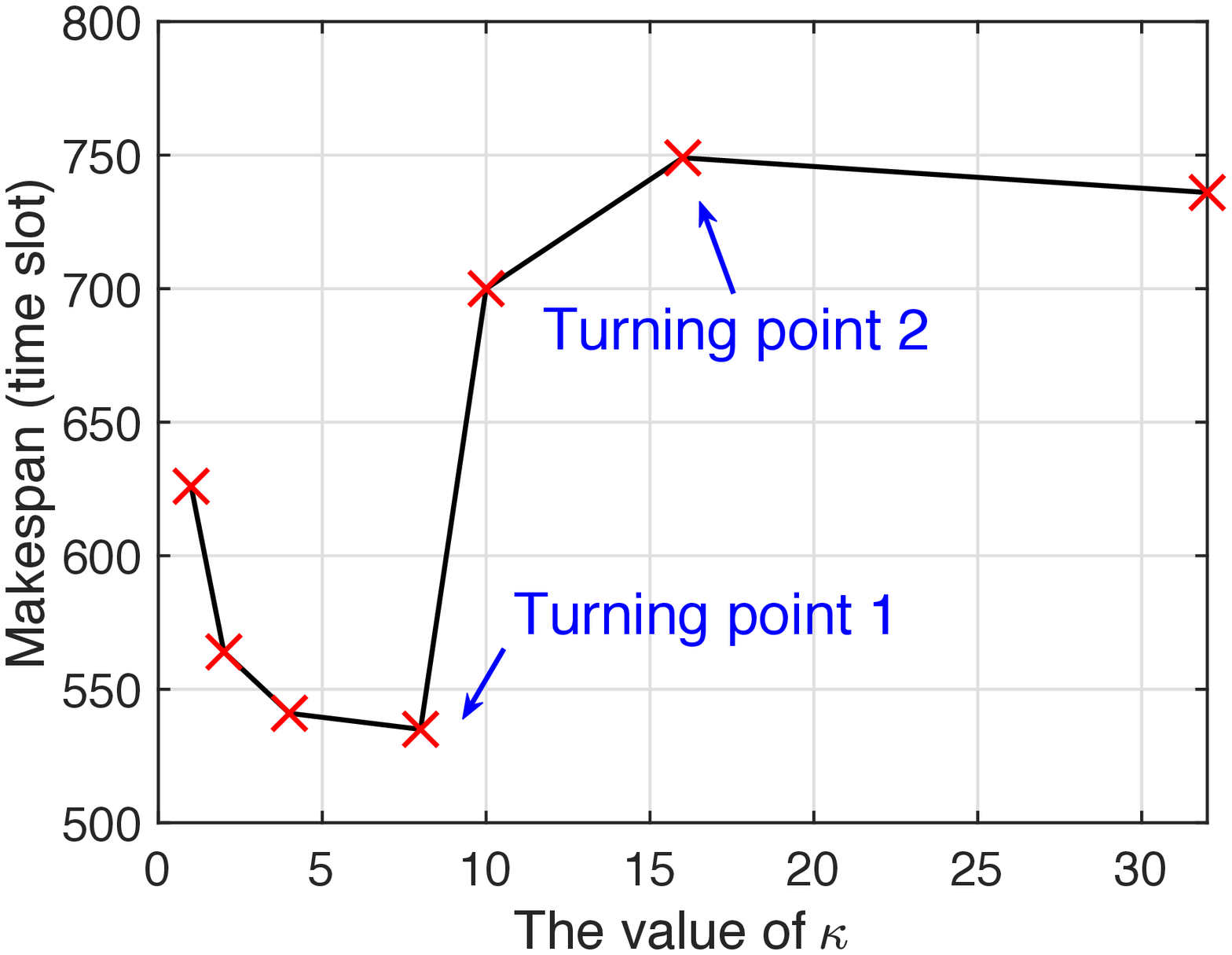}      
        \caption{Impact of value of $\kappa$ on makespan.} \label{fig:kappa}
    \end{minipage}%
      \hspace{0.01\textwidth}
         \begin{minipage}[t]{0.24\linewidth}
        \centering
        \includegraphics[width=1.04\textwidth]{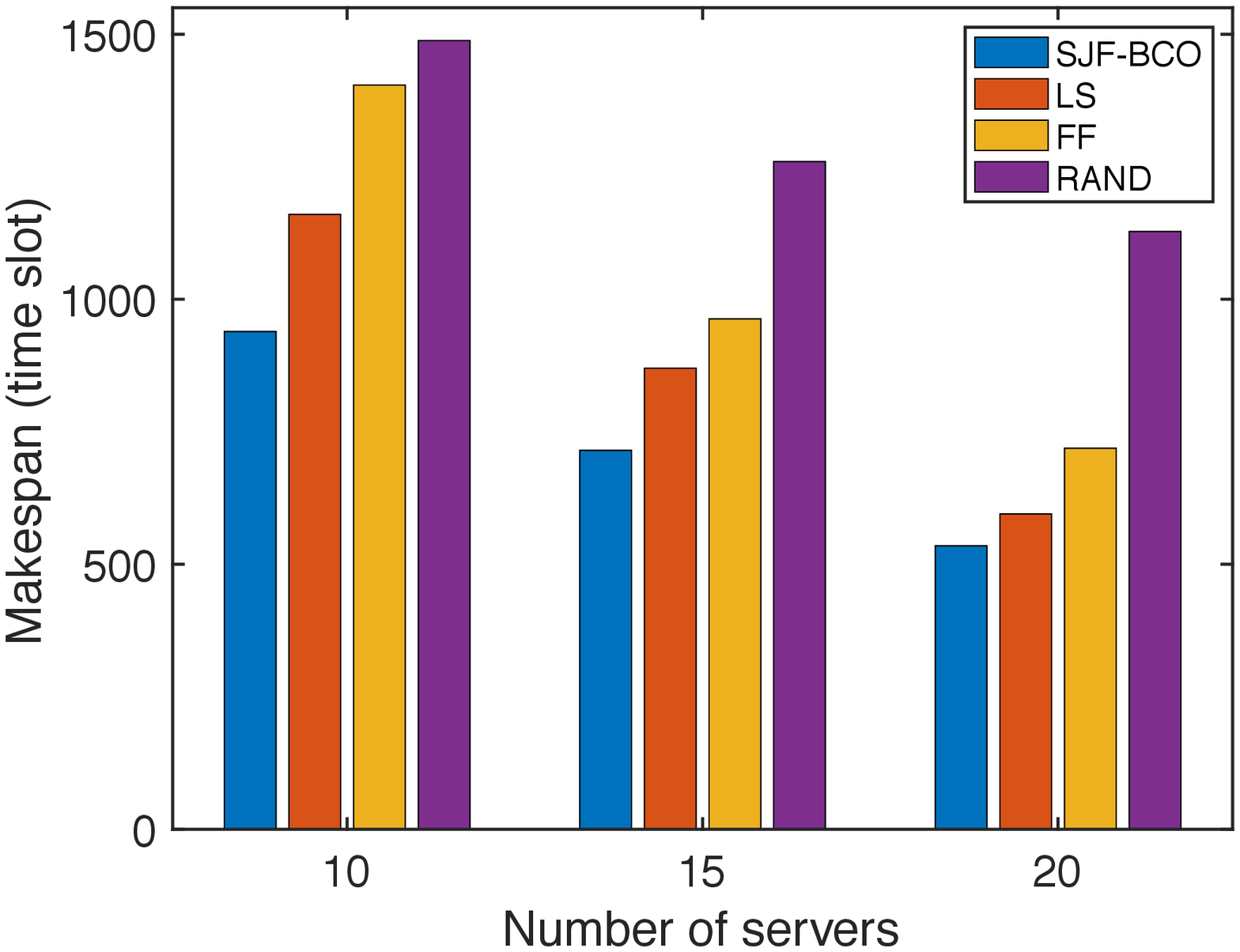}
        \caption{Makespan as the number of servers increases.} \label{fig:contention}
    \end{minipage}%
           \hspace{0.007\textwidth}
    \begin{minipage}[t]{0.24\linewidth}
        \centering
        \includegraphics[width=1.01\textwidth]{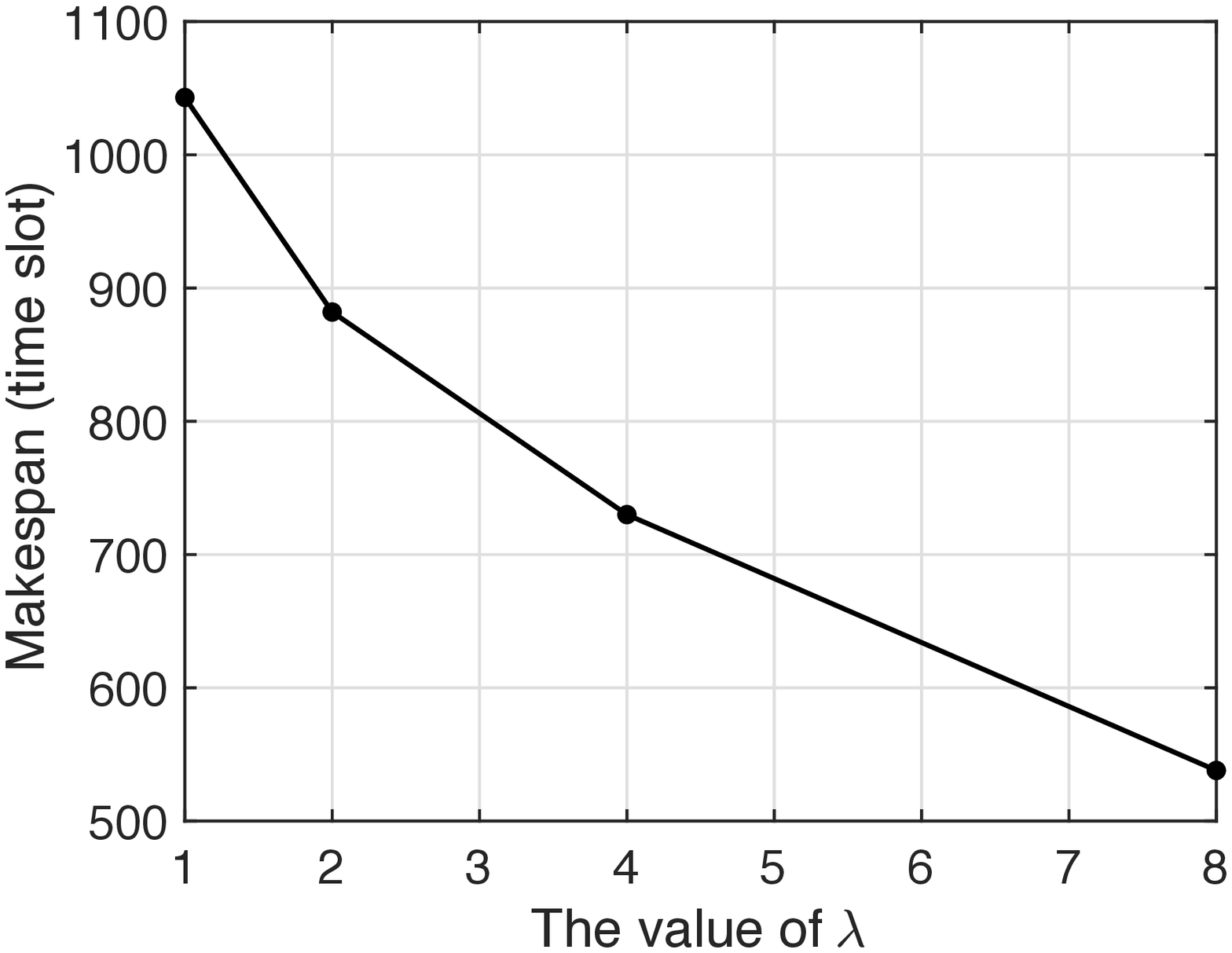}
        \caption{Impact of the value of $\lambda$ on makespan.}\label{fig:lambda}
    \end{minipage}%
\vspace{-.1in}
\end{figure*}

\section{Numerical Results}
\label{sec:numerical}

In this section, we conduct simulation studies to evaluate the performance of our proposed SJF-BCO algorithm.

\smallskip
\textbf{1) Experiment Settings:} 
Similar to the setting in~\cite{Wang20:contention}, the workload is generated  based on the Microsoft job trace~\cite{Jeon19:microsofttrace}.
We generate 160 DDL jobs by scaling down the original job trace~\cite{Jeon19:microsofttrace} following the job-type distribution, where there are 80 single-GPU jobs, 14 2-GPU jobs, 26 4-GPU jobs, 30 8-GPU jobs, 8 16-GPU jobs, and 2 32-GPU jobs.
We set $F_j\in[1000,6000]$.
The extra time cost brought by communication contention and overhead is within 15\% of the total actual execution time.
We let $\xi_1 = \xi_2$ (cf. Sec.~\ref{subsec:sys_model}) to make communication contention and overhead cost comparable.
We set $\tau_j[t]\in[0.01, 0.05]$~\cite{Yu21:PD-ORS}, and $\lambda_j = 1, \forall j$.
We set the estimated execution time $\hat{\rho}(\y^k)\in[50, 300]$ (evaluated from the product of $\tau_j[t]$ and $F_j$).
The GPU cluster has 20 servers.
The number of GPUs on each server is chosen from $\{4, 8, 16, 32\}$ uniformly at random.

\smallskip
\textbf{2) Baselines for Comparison:} We compare our algorithm with three representative job scheduling algorithms: First-Fit (FF)~\cite{Stavrinides11:FF-LS}, List-Scheduling (LS)~\cite{Stavrinides11:FF-LS}, and Random (RAND)~\cite{Wang20:contention}.
Here, we define $\theta_u^{f}$ as the maximum execution time limit returned by the scheduling policy $f$. 
Given a job $j$,
FF picks the first $G_j$ available GPUs such that their accumulative execution time does not exceed the limit $\theta_u^{FF}$, from server to server.
This policy tends to pack different jobs into the fewest number of servers to avoid fragmentation introduced by small jobs, which can save space for large jobs to be scheduled next.
LS selects top-$G_j$ GPUs with least execution time first, so that the accumulative execution time does not exceed the limit $\theta_u^{LS}$.
Note that this policy may introduce high communication overhead since it may choose GPUs from a large number of servers.
Further, LS tries to balance the execution time between GPUs by always selecting the one with the least execution time.
RAND randomly chooses servers and GPUs to schedule jobs.
In this policy, we allocate GPUs to a job as long as it does not exceed $T$, i.e., we set $\theta_u^{RAND} = T$, to avoid the long running time in order to find a feasible schedule.

\smallskip
\textbf{3) Experiment Results:}
First, we compare the makespan performance achieved by our SJF-BCO algorithm with those of the baseline  policies.
We set $T=1200$.
As shown in Fig.~\ref{fig:baseline}, SJF-BCO outperforms other scheduling policies both in terms of makespan and average job completion times, implying that SJF-BCO is also superior in terms of total job completion time.
Note that SJF-BCO tends to open new server(s) for large jobs to avoid the large communication overhead and use shared servers for small jobs to avoid the fragmentation, thus achieving better average completion time and makespan than FF and RAND.
Note that SJF-BCO has more prominent advantages over these baselines when the cluster has limited GPU resources.

Then, we examine the impact of $\kappa$ on the makespan in our proposed SJF-BCO algorithm.
We set $T=1200$, and select $\kappa$ from 1 to 32.
As indicated in Fig.~\ref{fig:kappa}, as the value of $\kappa$ increases, the makespan first drops and then increases and then drops again.
Recall that in Algorithm~\ref{alg:place}, FA-FFP is used when the number of requested GPUs $G_j\leq\kappa$; otherwise LBSGF is used.
Note that before Turning point 1 in Fig.~\ref{fig:kappa}, as $\kappa$ increases, the makespan drops since more small jobs are packed into the fewest number of shared servers, resulting in decrement of communication contention and overhead introduced by larger jobs to be scheduled later.
However, as $\kappa$ continues to grow, communication contention becomes more noticeable since more large jobs are scheduled to the shared servers, leading to the increase of makespan.
Finally, as $\kappa$ becomes sufficiently large, then the majority or even all jobs use shared servers to schedule their workers, which can slightly decrease the communication overhead due to the smaller resultant ring-span (see Turning point 2 in Fig.~\ref{fig:kappa}).


Next, we investigate the influence of communication contention by reducing the number of servers.
We set $T=1500$.
Intuitively, the larger number of servers, the less communication contention.
As we can see from Fig.~\ref{fig:contention}, as we increase the number of servers from 10 to 20, the makespan of FF, LS and SJF-BCO decrease due to the degradation of contention level.
Note that, if enough resources are available in the cluster, then each job will have a separate set of servers using SJF-BCO, i.e.,  its performance will become better as number of GPUs increases.
In this case, no communication contention will be introduced using SJF-BCO.
The intuition that FF has the largest makespan reduction is that the average idle time for workers drops dramatically since a smaller execution time limit could be set as the number of servers increases.

Lastly, we inspect the influence of $\lambda$ on the makespan for SJF-BCO with $\lambda \in\{1, 2, 4, 8\}$ and $\kappa=1$.
As we can see from Fig.~\ref{fig:lambda}, the makespan monotonically decreases as the $\lambda$ increases.
Recall that a larger $\lambda$-value implies a larger number of servers could be selected.
Then, the job has a higher chance to open new servers to schedule its workers, resulting in less communication contention and a smaller communication overhead.
Interestingly, $\lambda$ plays a similar role as $\kappa$, with the aim to balance communication overhead and contention.
Specifically, $\kappa$ affects the overall balance between all jobs since it determines the portion of jobs to use either FA-FFP or LBSGF, while $\lambda$ focuses more on the balance between communication contention and overhead for a specific job that uses LBSGF to schedule.


\section{Conclusion}
\label{sec:conclusion}

In this paper, we studied resource scheduling for DDL jobs in a multi-tenant GPU cluster, where we considered the communication contention and overhead determined by the distribution of workers.
We showed that this problem can be formulated as a highly non-trivial non-linear integer program with nonconvex and mixed packing-covering constraints.
We then converted the problem into a tractable integer linear program, which enables the design of approximation algorithms.
Specifically, we developed a new analytical model that jointly considers the placements and starting times of the workers of each DDL job.
Through careful reformulation, we then transformed the problem into an integer linear program with a more tractable structure, and proposed an approximation algorithm with an approximation ratio performance guarantee.
We provided rigorous theoretical analysis and conducted experiments to demonstrate the efficacy of our algorithms.
Collectively, our results contribute to a fundamental understanding on resource scheduling for DDL jobs in multi-tenant GPU clusters.

\section{Acknowledgements} \label{sec:ack}
J. Liu’s work has been supported in part by NSF grants CAREER CNS-2110259, CNS-2112471, CNS-2102233, CCF-2110252, and a Cisco Systems Research Grant GR127298.
B. Ji’s work has been supported in part by NSF CNS-2112694.
H. Rajan's work has been supported in part by NSF 21-20448 and NSF 19-34884.
\bibliographystyle{acm}
\bibliography{./BIB/DMLF,./BIB/ApproxAlg, ./BIB/CSMA}


\end{document}